\documentclass[a4paper, 11pt]{article}

\usepackage{fullpage}

\usepackage{amsmath, amssymb, amsthm}
\usepackage{graphicx}
\usepackage[font=small, labelfont=bf]{caption}
\usepackage{complexity}
\usepackage{myMath}
\usepackage[pdftex,colorlinks=true]{hyperref}

\usepackage{xcolor}
\definecolor{dark-red}{rgb}{0.4,0.15,0.15}
\definecolor{dark-blue}{rgb}{0.15,0.15,0.4}
\definecolor{medium-blue}{rgb}{0,0,0.5}
\definecolor{gray}{rgb}{0.5,0.5,0.5}

\hypersetup{
    colorlinks, linkcolor={dark-red},
    citecolor={dark-blue}, urlcolor={medium-blue}
}

\usepackage[linesnumbered, lined, boxed, noend]{algorithm2e}
\SetAlFnt{\small}
\SetAlCapFnt{\small}
\SetAlCapNameFnt{\small}

\usepackage[noblocks]{authblk}

\title{Fine-Grained Parameterized Complexity Analysis of
Graph Coloring Problems\footnote{This research was partially
funded by the Networks programme via the Dutch Ministry of Education, 
Culture and Science through the Netherlands Organisation for Scientific
Research. The second author was supported by NWO Veni grant ``Frontiers in
Parameterized Preprocessing''.}}

\author{Lars Jaffke\thanks{The research was done while this author was at
CWI, Amsterdam.}}
\affil{Department of Informatics\\
	University of Bergen\\
	Postboks 7803, N-5020 Bergen, Norway\\
  \texttt{l.jaffke@uib.no}}
  
\author{Bart M.\ P.\ Jansen}
\affil{Department of Mathematics and Computer Science\\
	Eindhoven University of Technology\\
	P.O. Box 513, 5600 MB Eindhoven, The Netherlands\\
  \texttt{b.m.p.jansen@tue.nl}}

\theoremstyle{plain}
\newtheorem{theorem}{Theorem}
\newtheorem{lemma}[theorem]{Lemma}
\newtheorem{proposition}[theorem]{Proposition}
\newtheorem{conjecture}[theorem]{Conjecture}
\newtheorem{observation}[theorem]{Observation}
\newtheorem*{propositionnonr}{Proposition}
\newtheorem{corollary}[theorem]{Corollary}
\theoremstyle{definition}
\newtheorem{definition}[theorem]{Definition}
\theoremstyle{remark}

\newtheorem*{remark*}{Remark}
\newtheorem*{claim*}{Claim}
\newtheorem*{observation*}{Observation}
\newtheorem*{openProblem*}{Open Problem}

\theoremstyle{plain}
\newtheorem{claim}[theorem]{Claim}
\newenvironment{claimproof}{\begin{proof}\renewcommand{\qedsymbol}{\claimqed}}{\end{proof}\renewcommand{\qedsymbol}{\plainqed}}

\let\plainqed\qedsymbol

\begin{document}

\maketitle

\abstract{The \textProb{$q$-Coloring} problem asks whether the vertices of a graph can be
	properly colored with $q$ colors. Lokshtanov et al. [SODA 2011] showed that
	\textProb{$q$-Coloring} on graphs with a feedback vertex set of size $k$ cannot
	be solved in time $\cO^*((q-\varepsilon)^k)$, for any $\varepsilon > 0$, unless
	the Strong Exponential-Time Hypothesis ($\SETH$) fails. In this paper we
	perform a fine-grained analysis of the complexity of \textProb{$q$-Coloring}
	with respect to a hierarchy of parameters. 
 	We show that even when parameterized by the vertex cover number, $q$ must
 	appear in the base of the exponent: Unless $\ETH$ fails, there is no universal 
 	constant $\theta$
 	such that \textProb{$q$-Coloring} parameterized by vertex cover can be solved
 	in time $\cO^*(\theta^k)$ for all fixed $q$.
	We apply a method due to Jansen and Kratsch [Inform. \& Comput. 2013] to prove
	that there are $\cO^*((q - \varepsilon)^k)$ time algorithms where $k$ is the
	vertex deletion distance to several graph classes $\cF$ for which \textProb{$q$-Coloring} is
	known to be solvable in polynomial time. 
	We generalize earlier ad-hoc results by showing that if $\cF$ is a class of
	graphs whose $(q+1)$-colorable members have bounded treedepth, then there
	exists some $\varepsilon > 0$ such that \textProb{$q$-Coloring} can be solved in time
	$\cO^*((q-\varepsilon)^k)$ when parameterized by the size of a given modulator
	to $\cF$.
	In contrast, we prove that if $\cF$ is the class of paths -- some of
	the simplest graphs of unbounded treedepth -- then no such algorithm can exist
	unless $\SETH$ fails.
}

\section{Introduction}
In an influential paper from 2011, Lokshtanov et al.\ showed that for several
problems, straightforward dynamic programming algorithms for graphs of bounded
treewidth are essentially optimal unless the Strong Exponential Time Hypothesis ($\SETH$) fails
\cite{LMS11}. (Section \ref{secExpTimeHyp} gives the definitions of the two Exponential Time Hypotheses, see \cite[Chapter 14]{CFKLMPPS15} or the survey \cite{Vas15} for further details.)
Some of the lower bounds, as the one for \textProb{$q$-Coloring}, even hold for
parameters such as the feedback vertex number, which form an upper bound on
the treewidth but may be arbitrarily much larger. For other problems such as
\textProb{Dominating Set}, the tight lower bound of
$\Omega((3-\varepsilon)^k)$ holds for the parameterization pathwidth, but is not
known for the parameterization feedback vertex set. In general, moving to a
parameterization that takes larger values might enable running times with a
smaller base of the exponent.
In this paper, we therefore investigate the parameterized complexity of the
\textProb{$q$-Coloring} and \textProb{$q$-List-Coloring} problems from a more
fine-grained perspective.
\par
In particular, we consider a hierarchy of graph parameters --- ordered by
their expressive strength --- which is a common method in parameterized
complexity, see e.g.\ \cite{FJR13} for an introduction. 
One of the strongest parameters for a graph problem is the number of vertices in a graph, 
in the following denoted by $n$. 
Bj\"{o}rklund et al.\ showed that the chromatic number $\chi(G)$ (the smallest
number of colors $q$ such that $G$ is $q$-colorable) of a graph $G$ can be
computed in time $\cO^*(2^n)$ \cite{BHK09}, so the base of the exponent in the runtime 
of the algorithm is independent of the value of $\chi(G)$. We show
that if you consider a slightly weaker parameter, the size $k$ of a vertex cover of
$G$, it is very unlikely that there is a constant $\theta$, 
such that \textProb{$q$-Coloring} can be solved in time $\cO^*(\theta^k)$
for all fixed $q \in \cO(1)$: It would imply that $\ETH$ is false. 
\par
However, we show that there is a simple algorithm that solves \textProb{$q$-Coloring} 
parameterized by vertex cover, and for which the base of the exponential in its runtime
is strictly smaller than the base $q$ that is potentially optimal for the treewidth 
parameterization.
(A proof of the following proposition is deferred to the beginning of Section \ref{secUB}.)
\begin{proposition}\label{thmVC3Col}
	There is an algorithm which decides whether a graph $G$ is $q$-colorable and runs in time $\cO^*((q-1.11)^k)$, where $k$ denotes the size of a \emph{given} vertex cover of $G$.
\end{proposition}
\par
On the other hand, the above algorithm does not obviously generalize to other
parameterizations. To derive more general results about obtaining non-trivial
runtime bounds for parameterized \textProb{$q$-Coloring}, we study graph classes
with small \emph{vertex modulators} to several graph classes $\cF$: Given a
graph $G = (V, E)$, a vertex modulator $X \subseteq V$ to $\cF$ is a subset of its
vertices such that if we remove $X$ from $G$ the resulting graph is a member of
$\cF$, i.e.\ $G - X \in \cF$. If $|X| \le k$, we say that $G \in \cF+kv$. (For
example, graphs that have a vertex cover of size at most $k$ are
$\mbox{\textsc{Independent}}+kv$ graphs.) Hence, we study the following
problems which were first investigated in this parameterized setting by Cai \cite{Cai03}.
\parproblemdef
	{$q$-(List-)Coloring on $\cF+kv$ Graphs}
	{An undirected graph $G$ and a modulator $X \subseteq V(G)$ such that $G - X
	\in \cF$ (and lists $\Lambda \colon V \to 2^{[q]}$).} 
	{$|X| = k$, the size of the modulator.}
	{Can we assign each vertex $v$ a color (from its list $\Lambda(v)$) such that
	adjacent vertices have different colors?}
Given a \textsc{No}-instance $(G, \Lambda)$ of \textProb{$q$-List-Coloring} we
call $(G', \Lambda')$ a \textsc{No}-subinstance of $(G, \Lambda)$, if $G'$ is an induced subgraph of $G$ and for all
vertices $v \in V(G')$: $\Lambda(v) = \Lambda'(v)$ such that $(G', \Lambda')$
is also \textsc{No}.
We show that if a graph class $\cF$ has small \textsc{No}-certificates for 
\textProb{$q$-List-Coloring} 
then \textProb{$q$-(List-)Coloring} on
$\cF+kv$ graphs can be solved in time $\cO^*((q-\varepsilon)^k)$, for some
$\varepsilon > 0$.
This notion was introduced by Jansen and Kratsch to prove the existence of polynomial
kernels for said parameterizations \cite{JK13}. \par
In addition to that, we give some further structural insight into hereditary graph classes
$\cF$, for which $\cF+kv$ graphs have non-trivial algorithms: We show that if
the $(q+1)$-colorable members of $\cF$ have bounded treedepth, then
$\cF+kv$ has $\cO^*((q-\varepsilon)^k)$ time algorithms for
\textProb{$q$-Coloring} when parameterized by the size $k$ of a given
modulator, for some $\varepsilon > 0$.
We prove that this \emph{treedepth-boundary} is in some sense tight: Arguably the most simple
graphs of unbounded treedepth are paths. We show that
\textProb{$q$-Coloring} cannot be solved in time $\cO^*((q-\varepsilon)^k)$ for
any $\varepsilon > 0$ on $\mbox{\textsc{Path}}+kv$ graphs, unless $\SETH$ fails
--- strengthening the lower bound for $\mbox{\textsc{Forest}}+kv$ graphs \cite{LMS11} via a somewhat simpler construction.
Using this strengthened lower bound, we prove that if a hereditary graph class $\cF$ excludes a complete bipartite graph $K_{t, t}$ for some constant $t$, then $\cF+kv$ has $\cO^*((q-\varepsilon)^k)$ time algorithms for \textProb{$q$-(List-)Coloring} \emph{if and only if} the $(q+1)$-colorable members of $\cF$ have bounded treedepth.
\par
The rest of the paper is organized as follows: In Section \ref{secPrel} we give
some fundamental definitions used throughout the paper. We present some upper
bounds in the hierarchy in Section \ref{secUB} and lower bounds in Section
\ref{secLB}. In Section \ref{secTD} we present the aforementioned tight relationship between the parameter treedepth and the existence of algorithms for \textProb{$q$-Coloring} with nontrivial runtime and we give concluding remarks in Section \ref{secConc}.

\section{Preliminaries}\label{secPrel}
We assume the reader to be familiar with the basic notions in
graph theory and parameterized complexity and refer to 
\cite{CFKLMPPS15, Die10, DF13, FG06} 
for an introduction. We now give the most important definitions which are
used throughout the paper. \par
We use the following notation: For $a, b \in \bN$ with $a < b$, $[a] =
\{1,\ldots,a\}$ and $[a..b] = \{a, a+1,\ldots,b\}$. The $\cO^*$-notation
suppresses polynomial factors in the input size $n$, i.e.\ $\cO^*(f(n,\cdot)) =
\cO(f(n, \cdot)\cdot n^{\cO(1)})$. For a function $f \colon X \to Y$, we denote by
$f_{\mid X'}$ the restriction of $f$ to $X' \subseteq X$.

\subsection{Graphs and Parameters}
Throughout the paper a graph $G$ with vertex set $V(G)$ and edge set $E(G)$ is
finite and simple. We sometimes shorthand $'V(G)'$ ($'E(G)'$) to $'V'$ ($'E'$)
if it is clear from the context.
For graphs $G$, $G'$ we denote by $G' \subseteq G$ that $G'$ is a subgraph of $G$,
i.e.\ $V(G') \subseteq V(G)$ and $E(G') \subseteq E(G)$. We often use the
notation $n = |V|$ and $m = |E|$. For a vertex $v \in V(G)$, we denote by $N_G(v)$ (or simply $N(v)$, if $G$ is clear from the context) the set of \emph{neighbors} of $v$ in $G$, i.e.\ $N_G(v) = \{w \in V(G) \mid \{v, w\} \in E(G)\}$.
\par
For a vertex set $V'
\subseteq V(G)$, we denote by $G[V']$ the subgraph \emph{induced} by $V'$, i.e.\
$G[V'] = (V', E(G) \cap V' \times V')$. A graph class $\cF$ is called
\emph{hereditary}, if it is closed under taking induced subgraphs. 
\par
We now list a
number of graph classes which will be important for the rest of the paper. 
A graph $G$ is \emph{independent}, if $E(G) = \emptyset$.
A \emph{cycle} is a connected graph all of whose vertices have degree two. A
graph is a \emph{forest}, if it does not contain a cycle as an induced subgraph
and a \emph{linear forest} if additionally its maximum degree is at most two.
A connected forest is a \emph{tree} and a tree of maximum degree at most two is a
\emph{path}.
A graph $G$ is a \emph{split graph}, if its vertex set $V(G)$ can be partitioned
into sets $W, Z \subseteq V(G)$ such that $G[W]$ is a clique and $G[Z]$ is
independent. We define the class \textsc{$\bigcup$Split}
containing all graphs that are disjoint unions of split graphs.
A graph $G$ is a \emph{cograph} if it does not contain $P_4$, a path on four vertices, as an 
induced subgraph. A graph is \emph{chordal}, if it does not have a cycle of length at least 
four as an induced subgraph. A cochordal graph is the edge complement of a chordal graph 
and the class \textsc{$\bigcup$Cochordal} contains all graphs that are disjoint 
unions of cochordal graphs.

\begin{definition}[Parameterized Problem]
	Let $\Sigma$ be an alphabet. A \emph{parameterized problem} is a set $\Pi
	\subseteq \Sigma^* \times \bN$, the second component being the 
	\emph{parameter} which usually expresses a structural measure of the input. 
	A parameterized problem is (strongly uniform)
	\emph{fixed-parameter tractable} (fpt) if there exists an algorithm to decide
	whether $\langle x, k\rangle \in \Pi$ in time $f(k)\cdot|x|^{\cO(1)}$
	where $f$ is a computable function.
\end{definition}

The main focus of our research is how the function $f(k)$ behaves for
\textProb{$q$-Coloring} w.r.t.\ different structural graph parameters, such as
the size of a vertex cover. 

In this paper we study a \emph{hierarchy of parameters}, a term which we will
now discuss. For a detailed introduction we refer to \cite[Section 3]{FJR13}.
For notational convenience, we denote by $\Pi_{p}$ a parameterized problem with
parameterization $p$.
Suppose we have a graph problem and two parameterizations $p(G)$ and
$p'(G)$ regarding some structural graph measure. We call parameterization
$p'(G)$ \emph{larger} than $p(G)$ if there is a function $f$, such that 
$f(p'(G)) \ge p(G)$ for all graphs $G$.
Modulo some technicalities, we can then observe that if a problem $\Pi_p$ is fpt, then
$\Pi_{p'}$ is also fpt. This induces a partial ordering on all parameterizations
based on which a hierarchy can be defined.

\subsection{Exponential-Time Hypotheses}\label{secExpTimeHyp}
In 2001, Impagliazzo et al.\ made two conjectures about the complexity of
\textProb{$q$-SAT} --- the problem of finding a satisfying assignment for a
Boolean formula in conjunctive normal form with clauses of size at most $q$
\cite{IP01, IPZ01}. These conjectures are known as the Exponential-Time
Hypothesis ($\ETH$) and Strong Exponential-Time Hypothesis ($\SETH$), formally
defined below. For a survey of conditional lower bounds based on such
conjectures, see \cite{Vas15}.
\par
\begin{conjecture}[$\ETH$ \cite{IP01}]
	There is an $\varepsilon > 0$, such that \textProb{3-SAT} on $n$ variables
	cannot be solved in time $\cO^*(2^{\varepsilon n})$.
\end{conjecture}
\begin{conjecture}[$\SETH$ \cite{IP01, IPZ01}]
	For every $\varepsilon > 0$, there is a $q \in \cO(1)$ such that \textProb{$q$-SAT} on $n$
	variables cannot be solved in time $\cO^*((2-\varepsilon)^n)$.
\end{conjecture}

\section{Upper Bounds}\label{secUB}
In this section we present upper bounds for parameterized
\textProb{$q$-Coloring}.
In particular, in Section \ref{secUBNoCert} we show that if a graph class $\cF$ has
\textsc{No}-certificates of constant size, then there exist
$\cO^*((q-\varepsilon)^k)$ time algorithms for \textProb{$q$-Coloring} on
$\cF+kv$ graphs for some $\varepsilon > 0$ depending on $\cF$.
In Section \ref{secUBTD} we show that if the $(q+1)$-colorable members of a hereditary
graph class $\cF$ have bounded treedepth, then $\cF$ has \textsc{No}-certificates of small size. 
\par
We begin by proving Proposition \ref{thmVC3Col} and repeat its statement.
\begin{propositionnonr}
	There is an algorithm which decides whether a graph $G$ is $q$-colorable and runs in time $\cO^*((q-1.11)^k)$, where $k$ denotes the size of a \emph{given} vertex cover of $G$.
\end{propositionnonr}
\begin{proof}
	Let $X \subseteq V(G)$ be the given vertex cover of $G$ of size $k$. We observe that if $G$ is $q$-colorable, then any valid $q$-coloring of $G$ can be extended from a valid $q$-coloring of $G[X]$. We know that in any $q$-coloring $\gamma \colon V \to [q]$ there is a color class that contains at most $\lfloor k/q \rfloor$ vertices in $X$. The algorithm now works as follows. We enumerate all sets $S \subseteq X$ of size at most $\lfloor k/q \rfloor$ and check whether they are independent. If so, let $S'$ denote the set consisting of $S$ together with all vertices in $V \setminus X$ that do not have a neighbor in $S$. Note that $G[S']$ is independent. We then recurse on the instance $G - S'$ with $q$ decreased by one (and the size of the modulator decreased by $|S|$). Once $q = 2$, we check whether the remaining graph is $2$-colorable (or equivalently, bipartite) in linear time.  
	
	We now compute the exponential dependence of the runtime by induction on $q$. As base cases we consider $q \in \{1,2,3\}$. The cases $q = 1$ and $q = 2$ are trivial, since the problem can be solved in polynomial time. For $q = 3$, the number of generated subproblems is bounded by
	$\sum_{\ell = 0}^{\lfloor k/3 \rfloor} {k \choose \ell}$,
which is at most $2^{H(1/3)k}$, where $H(x) = -x \log_2(x) - (1-x) \log_2(1-x)$ is the binary entropy \cite[page 427]{FG06}. Since $H(1/3) \le 0.9183$, the algorithm generates at most $2^{0.9183k} \le 1.89^k$ subproblems, all of which can be solved in polynomial time. For the induction step, let $q > 3$ and assume for the induction hypothesis that for $(q-1)$, the exponential dependence of the running time is upper bounded by $(q-1-1.11)^k$. Since the algorithm enumerates all subsets of $X$ of size $\ell$ for each $\ell \in [\lfloor k/q \rfloor]$, and the size of the parameter decreases by $\ell$ in each call, using the induction hypothesis we find that the exponential term in the running time is upper bounded by 
\begin{align*}
	\sum_{\ell = 0}^{\lfloor k/q \rfloor} {k \choose \ell}(q-2.11)^{k - \ell} \le \sum_{\ell = 0}^k {k \choose \ell}(q-2.11)^{k - \ell} \cdot 1^{\ell} = (q - 2.11 + 1)^k = (q - 1.11)^k,
\end{align*}
since $\sum_{i = 0}^n {n \choose i} \cdot a^i \cdot b^{n-i} = (a+b)^n$ by the Binomial Theorem.	
	
	We now argue the correctness of the algorithm, again by induction on $q$. The base cases, $q=1$ and $q=2$ are again trivially correct. For the induction step, consider $q > 2$ and assume for the induction hypothesis that the recursive calls to solve \textProb{$(q-1)$-Coloring} are correct. Suppose $G$ has a $q$-coloring $\gamma$ and let $T \subseteq V(G)$ denote the color class with the fewest vertices from $X$. Then, $|T \cap X| \le k/q$, so the algorithm guesses the set $S = T \cap X$. Since the corresponding set $S'$ contains all vertices in $G - X$ that do not have a neighbor in $S$ and $\gamma$ is a proper coloring, we can conclude that $S' \supseteq T$. Hence, $G - S'$ is a subgraph of the $(q-1)$-colorable graph induced by the other color classes of $\gamma$ which the algorithm detects correctly by the induction hypothesis. Conversely, any $(q-1)$-coloring for $G - S'$ can be lifted to a $q$-coloring of $G$ by giving all vertices in the independent set $S'$ the same, new, color.
\end{proof}

\subsection{Small No-Certificates}\label{secUBNoCert}
In earlier work \cite{JK13}, Jansen and Kratsch studied the kernelizability of
\textProb{$q$-Coloring} and established a generic method to prove the existence
of polynomial kernels for several parameterizations of \textProb{$q$-Coloring}.
We now show that we can use their method to prove the existence of
$\cO^*((q-\varepsilon)^k)$ time algorithms, for some $\varepsilon > 0$, for
several graph classes $\cF+kv$ as well.
\par 
We first introduce the necessary terminology. 
Let $(G, \Lambda)$ be an instance of \textProb{$q$-List-Coloring}. We call $(G',
\Lambda')$ a \emph{subinstance} of $(G, \Lambda)$, if $G'$ is an induced
subgraph of $G$ and $\Lambda(v) = \Lambda'(v)$ for all $v \in V(G')$.
\begin{definition}[$g(q)$-size \textsc{No}-certificates]
	Let $g \colon \bN \to \bN$ be a function. A graph class $\cF$ is said to have
	\emph{$g(q)$-size \textsc{No}-certificates} for \textProb{$q$-List-Coloring} if
	for all \textsc{No}-instances $(G, \Lambda)$ of \textProb{$q$-List-Coloring}
	with $G \in \cF$ there is a \textsc{No}-subinstance $(G', \Lambda')$ on at most
	$g(q)$ vertices.
\end{definition}

\begin{theorem}\label{thmAlgQListColNoCert}
	Let $\cF$ be a graph class with $g(q)$-size \textsc{No}-certificates
	for \textProb{$q$-List-Coloring}. Then, there is an $\varepsilon > 0$, such
	that \textProb{$q$-List-Coloring} (and hence, \textProb{$q$-Coloring}) on $\cF
	+ kv$ graphs can be solved in time $\cO^*((q-\varepsilon)^k)$ \emph{given} a modulator to $\cF$ of size at most $k$. 
	In particular, the algorithm runs in time $\cO^*\left(\sqrt[g(q)\cdot
	q]{q^{g(q) \cdot q} - 1}^k\right)$, where the degree of the hidden polynomial depends on $g(q)$.
\end{theorem}
\begin{proof}
	Let $G \in \cF + kv$ with vertex modulator $X$, such that $\cF$ has
	$g(q)$-size \textsc{No}-certificates for \textProb{$q$-List-Coloring}. 
	The idea of the algorithm is to enumerate partial colorings of $X$, except 
	some colorings for which it is clear that they cannot be extended to a proper
	coloring of the entire instance. The latter can occur as follows: After
	choosing a coloring for some vertices of $X$ and removing the chosen colors
	from the lists of their neighbors, a \textsc{No}-subinstance appears in the
	graph $G - X$. If the minimal \textsc{No}-subinstances have constant size, then for any
	given instance, either \emph{all} proper colorings on $X$ can be extended onto
	$G - X$, or there is a way to find a constant-size set $X' \subseteq X$
	of vertices for which at least one of the $q^{|X'|}$ colorings would trigger a
	\textsc{No}-subinstance and can therefore be discarded. Branching on the
	remaining relevant colorings for $X'$ then gives a nontrivial running time.
	An outline is given in Algorithm \ref{algQListColNoCert}.
	\par
	\begin{algorithm}
		\SetKwInOut{Input}{Input}
		\SetKwInOut{Output}{Output}
		
		\Input{A graph $G \in \cF + kv$ with vertex modulator $X$ and $\Lambda \colon V
		\to 2^{[q]}$.}
		\Output{\textsc{Yes}, if $G$ is $q$-list-colorable, \textsc{No} otherwise.}
		\BlankLine
		Let $\zeta$ be the set of \textsc{No}-instances of
		\textProb{$q$-List-Coloring} for $\cF$ of size at most $g(q)$, which is computed once by
		complete enumeration\;\label{algQLNCline1} 
		\If{\label{algQLNCMainCond} there exist $(H, \Lambda_H) \in \zeta$, $G' \subseteq G - X$
		and $X_1,\ldots,X_q \subseteq X$ of size at most $g(q)$ each such that:
		\begin{enumerate}
		  \item $\exists$ isomorphism $\varphi \colon V(G') \to V(H)$
		  \item For all $c \in [q]$ and $v \in X_c$ we have $c \in \Lambda(v)$
		  \item $(\forall v \in V(G'))(\forall c \in \Lambda(v) \setminus
		  \Lambda_H(\varphi(v)))$ $\exists w \in X_c$ with $\{v, w\} \in E(G)$\label{pty:neighbor}
		\end{enumerate}
		}{
			\ForEach{\label{algQLNCEnum}\text{proper coloring} $\gamma \colon \cX \to [q]$ where $\cX =
			\bigcup_i X_i$ and $\forall v \in \cX$: $\gamma(v) \in \Lambda(v)$}{ 
			\If{\label{algQLNCEnumExcl}$(\forall c \in
			[q])(\forall v \in X_c):~ \gamma(v) = c$}{ Skip this coloring, it is not 
				extendible to $G - X$\;
				} \Else {
					Create a copy $(G'', \Lambda'')$ of $(G, \Lambda)$ and denote by $\cX''$
					the vertex set in $G''$ corresponding to $\cX$ in $G$\;
					For each vertex $v \in \cX''$ and each neighbor $w$ of $v$: Remove
				  $\gamma(v)$ from $\Lambda''(w)$\;
				  Recurse on $(G'' - \cX'', \Lambda'')$\; \label{algQLNCRecurse}
				  \If{the recursive call returns \textsc{Yes}}{
						\textbf{Return} \textsc{Yes} and terminate the algorithm\;
					}
				}
			}
			\textbf{Return} \textsc{No}\; }
		\Else{Decide whether $(G[X], \Lambda)$ is $q$-list-colorable and
		if so, return \textsc{Yes}\;}\label{algQLNCline7}
		\caption{\textProb{$q$-List-Coloring} for $\cF + kv$ graphs where $\cF$ has
		$g(q)$-size \textsc{No}-certificates.}
		\label{algQListColNoCert}
	\end{algorithm}
	The main condition (line \ref{algQLNCMainCond}) checks
	whether the input graph $G$ contains the graph of a minimal
	\textsc{No}-instance as an induced subgraph.
	If so, we look for a neighborhood of $V(G')$ in $X$ (the sets $X_1,\ldots,X_q$), 
	which can block the colors that are on the lists $\Lambda$
	but not on the lists of the minimal \textsc{No}-instance.
	If these conditions
	are satisfied, then we know that we can exclude the coloring on
	$X_1,\ldots,X_q$ which assigns each vertex $v \in X_c$ the color $c$
	(for all $c \in [q]$): This coloring induces a \textsc{No}-subinstance on $(G,
	\Lambda)$. It suffices to use sets $X_c$ of at most $g(q)$ vertices each. To induce the \textsc{No}-instance, in the worst case we need a different vertex in $X_c$ for each of the $g(q)$ vertices in $H$ that do not have $c$ on their list.
	Hence, as described from line \ref{algQLNCEnum} on, we enumerate all
	colorings $\gamma \colon \cX \to [q]$ (where $\cX = \bigcup_i X_i$) except the one we just
	identified as not being extendible to $G - X$. For each such
	$\gamma$, we make a copy of the current instance and `assign' each vertex
	$v$ corresponding to a vertex in $\cX$ the color $\gamma(v)$: We remove
	$\gamma(v)$ from the lists of its neighbors and then remove $v$ from the copy
	instance. In the worst case we therefore recurse on $q^{q\cdot g(q)} - 1$
	instances with the size of the vertex modulator decreased by $q\cdot g(q)$.
	If during a branch in the computation, the condition in line \ref{algQLNCMainCond} is not
	satisfied, then we know that there is no coloring on the modulator that cannot
	be extended to the vertices outside the modulator and hence it is sufficient to
	compute whether $G[X]$ is $q$-list-colorable using the standard $\cO^*(2^n)$
	algorithm for computing the chromatic number \cite{BHK09}.
	As soon as one branch returns \textsc{Yes}, we can terminate the algorithm, since we found a valid list
	coloring.
	\par
	\begin{claim}\label{claim:only:gx}
	If the condition of line~\ref{algQLNCMainCond} does not hold, then~$G$ is $q$-list-colorable if and only if~$G[X]$ is $q$-list-colorable.
	\end{claim}
	\begin{claimproof}
	The forward direction is trivial since any proper coloring of~$G$ yields a proper coloring of its induced subgraph~$G[X]$. To prove the reverse direction, we show that if the condition of line~\ref{algQLNCMainCond} fails, any proper $q$-list-coloring of~$G[X]$ can be extended to a proper $q$-list-coloring of the entire graph. 
	\par
	Suppose that~$\gamma \colon X \to [q]$ is a proper $q$-list-coloring of~$G[X]$. Define a $q$-list-coloring instance~$(G - X, \Lambda')$ on the graph~$G - X$, where for each vertex~$v \in V(G - X)$ the list of allowed colors is~$\Lambda'(v) \defeq \Lambda(v) \setminus \{ \gamma(u) \mid u \in N_G(v) \cap X \}$. If~$(G - X, \Lambda')$ has a proper $q$-list-coloring~$\gamma'$, then we can obtain a proper $q$-list-coloring for~$G$ by following~$\gamma$ on the vertices in~$X$ and~$\gamma'$ on the vertices outside~$X$. The fact that the colors for vertices in~$X$ are removed from the $\Lambda'$-lists of their neighbors ensures that the resulting coloring is proper, and since each list of~$\Lambda'$ is a subset of the corresponding list in~$\Lambda$, the coloring satisfies the list requirements. We therefore complete the proof by showing that~$(G-X, \Lambda')$ must be a \textsc{Yes}-instance. Assume for a contradiction that~$(G-X, \Lambda')$ has answer \textsc{No}. Since~$G - X \in \mathcal{F}$, which has~$g(q)$-size \textsc{No}-certificates, there is an induced subinstance~$(G', \Lambda'')$ of~$(G-X, \Lambda')$ on at most~$g(q)$ vertices, where~$G'$ is an induced subgraph of~$G - X$ and therefore of~$G$. Since~$(G', \Lambda'')$ is a \textsc{No}-instance on at most~$g(q)$ vertices, the instance~$(H \defeq G', \Lambda_H \defeq \Lambda'')$ is contained in the set of enumerated small \textsc{No}-instances. For each~$v \in V(G')$, for each color~$c$ that belongs to~$\Lambda(v)$ but not to~$\Lambda'(v) = \Lambda''(v)$ we have~$\gamma(u) = c$ for some~$u \in N_G(v) \cap X$, by definition of~$\Lambda'$. Initialize~$X_1, \ldots, X_q$ as empty vertex sets. For each~$v \in V(G')$ and color~$c \in \Lambda(v) \setminus \Lambda''(v)$, add such a vertex~$u$ to~$X_c$. Since~$\gamma$ satisfies the list constraints, for each vertex~$v \in X_c$ with~$c \in [q]$ we have~$c \in \Lambda(v)$. Hence these structures satisfy the conditions of line~\ref{algQLNCMainCond}; a contradiction.
\end{claimproof}
	\par
	\begin{claim}\label{claim:skip:one:coloring}
	If the condition of line~\ref{algQLNCEnumExcl} holds, then the coloring~$\gamma$ cannot be extended to a proper $q$-list-coloring of~$G$.
	\end{claim}
	\begin{claimproof}
	To extend the coloring~$\gamma$ to the entire graph~$G$, each vertex~$v$ of~$G-X$ has to receive a color of~$\Lambda(v) \setminus \{ \gamma(u) \mid u \in N_G(v) \cap X\}$, since the color of~$v$ must differ from that of its neighbors. For each vertex~$v$ in the subgraph~$G'$, for each color~$c$ in~$\Lambda(v) \setminus \Lambda_H(\phi(v))$ there is a neighbor of~$v$ in~$X_c$ (by condition~\ref{pty:neighbor} of line~\ref{algQLNCMainCond}) that is colored~$c$ (by line~\ref{algQLNCEnumExcl}). Hence the colors available for~$v$ in an extension form a subset of~$\Lambda_H(\phi(v))$. But since~$G'$ is isomorphic to~$H$, and~$(H, \Lambda_H)$ is a \textsc{No}-instance, no such extension is possible as it would yield a proper $q$-list-coloring of~$(H, \Lambda_H)$.
\end{claimproof}
	Using these claims we prove correctness by induction on the nesting depth of recursive calls in which the condition of line~\ref{algQLNCMainCond} is satisfied. If line~\ref{algQLNCMainCond} is not satisfied (which includes the base case of the induction), then the algorithm is correct by Claim~\ref{claim:only:gx} and the fact that we invoke a correct algorithm in line \ref{algQLNCline7} as a subroutine~\cite{BHK09}. Now, suppose that the condition of line~\ref{algQLNCMainCond} is satisfied, and assume by the induction hypothesis that the recursive calls (line \ref{algQLNCRecurse}) are correct. Let $(G,
	\Lambda)$ with modulator $X$ be the current instance. We recurse on each possible proper $q$-list-coloring of the set $\cX$, except
	the one described in the condition in line \ref{algQLNCEnumExcl} for which Claim~\ref{claim:skip:one:coloring} shows it cannot be extended to a proper $q$-list-coloring. If~$(G,\Lambda)$ has a proper $q$-list-coloring~$\gamma$, then in the branch where we correctly guess the restriction of~$\gamma$ onto the vertices in~$\cX$ we find a \textsc{Yes}-answer: the restriction of~$\gamma$ on~$G'' - \cX''$ is a proper $q$-list-coloring of~$(G'' - \cX, \Lambda'')$ since the colors we removed from the lists were not used on~$G'' - \cX''$ (they were used on their neighbors in~$\cX''$). Conversely, if some recursive call yields a \textsc{Yes}-answer, then since we restricted the lists before going into recursion, we can extend a proper $q$-list-coloring on the smaller instance with the coloring~$\gamma$ on~$\cX$ to obtain a proper $q$-list-coloring of~$(G,\Lambda)$.
	\par
	We now analyze the runtime. Since $q$ is a constant, $g(q)$ is constant as well
	and computing the set $\zeta$ in line 1 can be done in
	constant time. Using the same argument we observe that the condition in line
	\ref{algQLNCMainCond} checks a polynomial number of options: The size of $\zeta$ and the size of its
	elements are constant and hence there is a polynomial number (at most~$|\zeta| \cdot n^{g(q)}$) of subgraphs of
	$G$ to consider. Since $t \le g(q)$, we can enumerate all isomorphisms and all sets
	$X_1,\ldots,X_q$ with an additional polynomial overhead. Hence the work in each iteration, excluding the recursive calls and line~\ref{algQLNCline7}, is polynomial.
	\par
	Line \ref{algQLNCline7} can be done in time~$\cO^*(2^{k + q})$, which is~$\cO^*(2^k)$ for constant~$q$, using the~$\cO^*(2^n)$ algorithm for \textProb{Chromatic Number}~\cite{BHK09} and the following classic reduction from $q$-list-coloring to $q$-coloring. Instance~$(G[X],\Lambda)$ has a proper $q$-list-coloring if and only if the following graph is $q$-colorable: starting from~$G[X]$, add a $q$-clique whose vertices represent the~$q$ colors, and edges between every~$v \in X$ and the clique-vertices whose colors do not appear on~$\Lambda(v)$.
	\par
	Using these facts we bound the total runtime. In the worst case we branch on $q^{q\cdot g(q)} -
	1$ instances in which the size of the modulator decreased by $q \cdot g(q)$. By standard techniques~\cite[Proposition 8.1]{Niedermeier06}, this branching vector can be shown to generate a search tree with~$\cO(\sqrt[g(q)\cdot	q]{q^{g(q) \cdot q} - 1}^k)$ nodes. If the work at each node of the tree is polynomial, we therefore get a total runtime bound matching the theorem statement. If we do not execute line~\ref{algQLNCline7}, then indeed a single iteration takes polynomial time. If line~\ref{algQLNCline7} is executed, then we spend~$\cO^*(2^k)$ time on the iteration. However, in that case we do not recurse further, so the time spent solving the problem on~$G[X]$ can be discounted against the fact that we do not explore a search tree of size~$\sqrt[g(q)\cdot	q]{q^{g(q) \cdot q} - 1}^k > 2^k$ for~$q \geq 3$. The time bound follows.
	\par
	This concludes the proof of Theorem \ref{thmAlgQListColNoCert}, noting that we can 
	apply any algorithm for \textProb{$q$-List-Coloring} to solve an instance of 
	\textProb{$q$-Coloring} by giving each vertex in a given instance of 
	\textProb{$q$-Coloring} a full list.
\end{proof}
In the light of \cite[Lemmas 2-4]{JK13} we can apply Theorem \ref{thmAlgQListColNoCert}
to a number of graph classes.
\begin{corollary}[of Thm. \ref{thmAlgQListColNoCert} and 
Cor. 1 and 2 and Lemmas 2, 3 and 4 in 
\cite{JK13}]\label{corAlgQListColNoCert}
	There is an $\varepsilon > 0$, such that the \textProb{$q$-Coloring} and \textProb{$q$-List-Coloring} problems on $\cF+kv$ graphs can be solved in time $\cO^*((q-\varepsilon)^k)$ \emph{given} a modulator to $\cF$ of size $k$, where $\cF$ is one of the following classes: \textsc{Independent, $\bigcup$Split, $\bigcup$Cochordal} and \textsc{Cograph}.
\end{corollary}
\begin{remark*}
	Rather than on the maximum \emph{size} of any minimal \text{No}-instance of a
	graph class $\cF$, the runtime of the algorithm described in Theorem
	\ref{thmAlgQListColNoCert} depends on their maximum \emph{deficiency}, defined
	as $d(G, \Lambda) = \sum_{v \in V(G)}q - |\Lambda(v)|$ (as we need one vertex
	in the modulator for each color we want to block from the list of a vertex in
	the \textsc{No}-instance).
	We would like to note that the runtime of Algorithm \ref{algQListColNoCert} can be improved when
	analyzing the deficiency of the \textsc{No}-instances constructed in the proofs
	of \cite[Lemmas 2-4]{JK13}.
\end{remark*}

\subsection{Bounded Treedepth}\label{secUBTD}
We now show that if the $(q+1)$-colorable members of a hereditary graph class $\cF$ have
treedepth at most $t$, then $\cF$ has $q^t$-size \textsc{No}-certificates. For a detailed introduction to the parameter treedepth and its applications, we refer to \cite[Chapter 6]{NO12}.
\begin{definition}[Treedepth]
	Let $G$ be a connected graph. A \emph{treedepth decomposition} $\cT = (V(G),
	F)$ is a rooted tree on the vertex set of $G$ such that the following holds.
	For $v \in V(G)$, let $\cA_v$ denote the set of ancestors of $v$ in $\cT$. Then, for each edge
	$\{v, w\} \in E(G)$, either $v \in \cA_w$ or $w \in \cA_v$. \par
	The \emph{depth} of $\cT$ is the number of vertices on a longest path from
	the root to a leaf. The \emph{treedepth} of a connected graph is the minimum depth of
	all its treedepth decompositions. The treedepth of a disconnected graph is
	the maximum treedepth of its connected components.
\end{definition}
The main result of this section is the following.
\begin{lemma}\label{lemQLColUBTDepth}
	Let $\cF$ be a hereditary graph class whose $(q+1)$-colorable members have
	treedepth at most $t$. Then, $\cF$ has $q^t$-size \textsc{No}-certificates for
	\textProb{$q$-List-Coloring}.
\end{lemma}
\begin{proof}
	Consider an arbitrary \textsc{No}-instance $(G, \Lambda)$ of
	\textProb{$q$-List-Coloring} for a graph $G \in \cF$. If $G$ is not
	$(q+1)$-colorable (ignoring the lists $\Lambda$), then remove an arbitrary
	vertex from $G$. Since this lowers the chromatic number by at most one, the
	resulting graph will still be a \textsc{No}-instance of \textProb{$q$-Coloring} and therefore of
	\textProb{$q$-List-Coloring}. Repeat this step until arriving at a subinstance
	$(G', \Lambda')$ that is $(q+1)$-colorable. By assumption, $G'$ has treedepth
	at most $t$. Fix an arbitrary treedepth decomposition for $G'$ of depth at most
	$t$. We use the decomposition to find a \textsc{No}-subinstance by a 
	recursive algorithm. Given a \textsc{No}-instance $(G,
	\Lambda)$ and a treedepth decomposition $\cT$ of $G$ of depth at most $t$, it
	marks a set $M \subseteq V(G)$ such that the subinstance induced by $M$ is
	still a \textsc{No}-instance and $|M| \le q^t$.
	\par
	If the treedepth decomposition has depth one, then mark a vertex with an
	empty list (which must exist if the answer is \textsc{No}). When the
	decomposition has depth $> 1$, then do the following. Let $\cT$ be a tree of
	the decomposition that represents a connected component $C$ that cannot be list
	colored. Let $r$ be its root. For each color $c \in \Lambda(r)$, create a list
	coloring instance~$(C - \{r\}, \Lambda_c)$ on a graph of treedepth $t-1$ as follows. The graph is $C -
	\{r\}$ and its decomposition consists of $\cT$ minus its root (which therefore
	splits into a forest), and the lists equal the old lists except that we remove
	$c$ from the lists of all of $r$'s neighbors. Observe that the subinstance has
	answer \textsc{No}, since otherwise the component $C$ has a proper coloring.
	Recursively call the algorithm on this smaller instance to get a set $M_c$ that
	preserves the fact that $(C - \{r\}, \Lambda_c)$ has answer \textsc{No}.
	After getting the answers from all the recursive calls, mark the vertices in
	the set $M$ containing the root $r$ together with the union of the sets $M_c$
	for all $c \in \Lambda(r)$.
	\par
	To bound the size of the set $M$, let $h(t)$ denote the maximum number
	of marked vertices in a treedepth decomposition of depth $t$. Clearly, $h(1) =
	1$. If $t > 1$, we recurse in at most $q$ ways on instances of treedepth $t-1$,
	hence the number of marked vertices is described by the recurrence $h(t) \le
	q\cdot h(t-1) + 1$ which resolves to $h(t) \le \frac{q^t - 1}{q-1}$ and hence
	$h(t) \le q^t$, as claimed.
	\par
	We now prove that the above described marking procedure preserves the
	\textsc{No}-answer of an instance of \textProb{$q$-List-Coloring}. We use
	induction on $t$, the depth of a treedepth decomposition $\cT$ (with root
	$r$) of the graph $G$ of a \textProb{$q$-List-Coloring} \textsc{No}-instance
	$(G, \Lambda)$. 
	The base case $t = 1$ is trivially correct: A graph has treedepth one if and 
	only if it is independent and since a graph is $q$-list-colorable if and only if its
	connected components are $q$-list-colorable, the only minimal
	\textsc{No}-instance of treedepth one is a single vertex with an empty list,
	which we marked in the procedure. Now suppose for the induction hypothesis that
	$t > 1$ and for all $t' < t$, the marking procedure is correct. 
	Consider a treedepth decomposition $\cT$ of a connected component $C$ of (a
	subgraph of) $G$ and the set $M$ of currently marked vertices. 
	Suppose for the sake of a contradiction that $(G[M], \Lambda_{\mid M})$ is a
	\textsc{Yes}-instance with proper list-coloring $\gamma \colon M \to [q]$. Let
	$C_{\gamma(r)}$ denote the connected component of $C - \{r\}$ we branched on
	for color $\gamma(r)$ and $M_{\gamma(r)}$ the set of marked vertices in
	$C_{\gamma(r)}$.	
	By the induction hypothesis (which applies since $C_{\gamma(r)}$ has treedepth
	at most $t-1$), we know that $(G[M_{\gamma(r)}], \Lambda_{\gamma(r)})$ is a \textsc{No}-instance of
	\textProb{$q$-List-Coloring}. But~$\gamma_{\mid M_{\gamma(r)}}$ is a valid solution for that instance if~$\gamma$ is a proper coloring: the color of~$r$ cannot appear on its neighbors in~$M_{\gamma(r)}$, and therefore~$\gamma_{\mid M_{\gamma(r)}}$ satisfies the list constraints of~$\Lambda_{\gamma(r)}$. This contradicts the fact that~$(g[M_{\gamma(r)}], \Lambda_{\gamma(r)})$ is a \textsc{No}-instance.
\end{proof}

To see the versatility of Lemma \ref{lemQLColUBTDepth}, observe that the
vertices of a $(q+1)$-colorable split graph can be partitioned into a clique of
size at most $(q+1)$ and an independent set, which makes it easy to see that
they have treedepth at most $q+2$. Since the treedepth of a disconnected graph
equals the maximum of the treedepth of its connected components, we then get a finite
($q^{q+2}$) bound on the size of minimal \textsc{No}-instances for
\textProb{$q$-List-Coloring} on $\bigcup\mbox{\textsc{Split}}$ graphs. An ad-hoc
argument was needed for this in earlier work \cite[Lemma 2]{JK13}, albeit
resulting in a better bound ($q + 4^q$).

\section{Lower Bounds}\label{secLB}
In this section we prove lower bounds for \textProb{$q$-Coloring} in the parameter hierarchy. 
Since in the following, the `$\cF+kv$'-notation is more convenient for the presentation of our results, we will mostly refer to graphs which have a vertex cover of size $k$ as $\mbox{\textsc{Independent}}+kv$ graphs and graphs that have a feedback vertex set of size $k$ as $\mbox{\textsc{Forest}}+kv$ graphs.

In Section \ref{secLBETHVC} we show that there is no universal constant $\theta$, such that \textProb{$q$-Coloring} on $\mbox{\textsc{Independent}} + kv$ graphs can be solved in time $\cO^*(\theta^k)$ for all fixed $q \in \cO(1)$, unless $\ETH$ fails. We generalize the lower bound modulo $\SETH$ for $\mbox{\textsc{Forest}}+kv$ graphs \cite{LMS11} to $\mbox{\textsc{Linear Forest}}+kv$ (and $\mbox{\textsc{Path}}+kv$) graphs in Section \ref{secLBSETHLinFor}. Note that by the constructions we give in their proofs, the lower bounds also hold in case a modulator of size $k$ to the respective graph class is given.

\subsection{No Universal Constant for Independent+kv graphs}\label{secLBETHVC} 
The following theorem shows that, unless $\ETH$ fails, the runtime of any
algorithm for \textProb{$q$-Coloring} parameterized by vertex cover
(equivalently, on $\mbox{\textsc{Independent}}+kv$ graphs), always has a term
depending on $q$ in the base of the exponent.

\begin{theorem}\label{thmNoUCQColETH}
	There is no (universal) constant $\theta$, such that for all fixed $q \in
	\cO(1)$, \textProb{$q$-Coloring} on $\mbox{\textsc{Independent}}+kv$ graphs can
	be solved in time $\cO^*(\theta^k)$, unless $\ETH$ fails.
\end{theorem}
\begin{proof}
	Assume we can solve \textProb{$q$-Coloring} on $\mbox{\textsc{Independent}}+kv$
	graphs in time $\cO^*(\theta^k)$. We will use this hypothetical algorithm to solve
	\textProb{3-SAT} in $\cO^*(2^{\varepsilon n})$ time for arbitrarily small
	$\varepsilon > 0$, contradicting $\ETH$. We present a way to reduce an instance
	$\varphi$ of \textProb{3-SAT} to an instance of \textProb{$3q$-List-Coloring}
	for $q$ an arbitrary power of $2$. The larger $q$ is, the smaller the vertex
	cover of the constructed graph will be. It will be useful to think of a color
	$c \in [q]$ ($q = 2^t$ for some $t \in \bN$) as a bitstring of length $t$, which naturally encodes a truth assignment
	to $t$ variables. The entire color range~$[3q]$ partitions into three consecutive blocks of~$q$ colors, so that the same truth assignment to~$t$ variables 
	can be encoded by three distinct colors~$c, c+q$, and~$c+2q$ for some~$c \in [q]$. 
	The reason for the threefold redundancy is that clauses in
	$\varphi$ have size three and will become clear in the course of the proof.
	\par
	Given an instance $\varphi$ of \textProb{3-SAT}, we create a graph $G_{3q}$
	and lists $\Lambda \colon V(G_{3q}) \to [3q]$ as follows. First, we add $\lceil
	n/\log q \rceil$ vertices $v_{1,i}$ (where $i \in [\lceil n/\log q \rceil]$) 
	to $V(G_{3q})$, whose colorings will correspond to the truth
	assignments of the variables~$x_1, \ldots, x_n$ in $\varphi$. We let $\Lambda(v_{1, i}) = [q]$ for all these vertices. In
	particular, the variable $x_i$ will be encoded by vertex $v_{1, \lceil i/\log q
	\rceil}$.
	We add two more layers of vertices $v_{2, i}, v_{3, i}$ (where 
	$i \in [\lceil n / \log q \rceil]$) to $G_{3q}$ whose lists will be $\Lambda(v_{2,
	i}) = [(q+1)..2q]$ and $\Lambda(v_{3, i}) = [(2q + 1)..3q]$, respectively (for
	all $i$). Throughout the proof, we denote the set of all these \emph{variable
	vertices} by $\cV = \bigcup_{i, j} v_{i, j}$, where $i \in [3]$ and $j \in
	[\lceil n/\log q \rceil]$.
	\par
	For each $i \in [2]$ and $j \in [\lceil n/\log q \rceil]$ we do the following.
	For each pair of colors $c \in [((i-1)q + 1)..(i\cdot q)]$ and $c' \in
	[(i\cdot q + 1)..((i+1)q)]$ such that $c + q \neq c'$, we add a vertex $u^{i,j}_{c,
	c'}$ with list $\Lambda(u^{i,j}_{c, c'}) = \{c, c'\}$ and make it adjacent to both
	$v_{i, j}$ and $v_{i+1, j}$.
	Note that this way, we add $\cO(q^2)$ and hence a constant number of vertices for each such
	$i$ and $j$. We denote the set of all vertices $u^{\cdot,\cdot}_{\cdot, \cdot}$ for all $i$
	and $j$ by $\cU$.
	\begin{claim}\label{obsNoUCVCDelta}
		Let $i \in [2]$ and $j \in [\lceil n/\log q \rceil]$. In any proper
		list-coloring of $G_{3q}$, the color $c \in [((i-1)q + 1)..(i\cdot q)]$
		appears on $v_{i, j}$ if and only if the color $c + q$ appears on $v_{i+1,
		j}$. If color~$c \in [((i-1)q + 1)..(i\cdot q)]$ appears on~$v_{i,j}$ and~$c' = q + c$ appears on~$v_{i+1,j}$, then all vertices~$u^{i,j}_{\cdot, \cdot}$ can be 
		assigned a color from their list that does not appear on a neighbor.
	\end{claim}
	\begin{claimproof}
		We first observe that the lists of $v_{i, j}$ and $v_{i+1, j}$ are
		$\Lambda(v_{i, j}) = [((i-1)q+1)..(i\cdot q)]$ and $\Lambda(v_{i+1, j}) =
		[(i\cdot q + 1)..(i+1)q]$, respectively. Suppose that $c$ appears on $v_{i,
		j}$. Then, for every color $c' \in [(i\cdot q + 1)..((i+1)q)]$ with $c' \neq c +
		q$ there is a neighbor $u^{i,j}_{c, c'}$ of $v_{i, j}$ with list $\Lambda(u^{i,j}_{c,
		c'}) = \{c, c'\}$. Since $c$ already appears on a neighbor of $u^{i,j}_{c, c'}$, 
		we know that in each proper coloring, $u^{i,j}_{c, c'}$ must be
		colored $c'$, blocking this color for its neighbor~$v_{i+1,j}$. As this prevents any color~$c' \neq c + q$ from appearing on~$v_{i+1,j}$, in any proper list-coloring that vertex is colored~$c+q$. 
		(A proof of the converse works the same way.)
 		\par 
		Now suppose~$c$ appears on~$v_{i,j}$ and~$c+q$ appears on~$v_{i,j}$. Then any vertex~$u^{i,j}_{c', c''}$ created by the process above has~$\{c', c''\} \neq \{c, c+q\}$ by construction. Hence~$u^{i,j}_{c',c''}$ can safely be assigned a color of~$\{c',c''\} \setminus \{c, c+q\}$, which does not appear on any of its neighbors.
	\end{claimproof}
	Claim \ref{obsNoUCVCDelta} shows that in any proper list-coloring of~$\cV$, 
	there is a threefold redundancy: If color $c$
	appears on $v_{1, i}$, then color $c + q$ appears on $v_{2, i}$ and $c + 2q$
	appears on $v_{3, i}$. We associate a proper list-coloring of~$\cV$ with the truth assignment whose \textsc{True/False}
	assignment to the $i$-th block of $\log q$ consecutive
	variables follows the $1/0$-bit pattern in the least significant $\log q$ bits of the
	binary expansion of the color of vertex $v_{1, i}$. Conversely, given a truth assignment to~$x_1, \ldots, x_n$ we associate it to the coloring of~$\cV$ where 
	the color of vertex~$v_{1, i}$ is given by the number whose least significant~$\log q$ bits match the truth assignment to the $i$-th block of~$\log q$ variables, and any remaining bits are set to~$0$. The colors of~$v_{2,i}$ and~$v_{3,i}$ are~$q$ and~$2q$ higher than the color of~$v_{1,i}$.
	\par
	For each clause $C_j \in \varphi$ we will now add a
	number of \emph{clause vertices} to ensure that if
	$C_j$ is not satisfied by a given truth assignment of its variables, 
	then the corresponding coloring of
	the vertices $\cV$ cannot be extended to (at least) one of these clause vertices.
	\par Let $C_j \in \varphi$ be a clause with variables $x_{j_1}, x_{j_2}$, and
	$x_{j_3}$. Then, $v_{1, \lceil j_1/\log q \rceil}, v_{1, \lceil j_2/\log
	q\rceil}$, and $v_{1, \lceil j_3/\log q \rceil}$ denote the vertices whose
	colorings encode the truth assignments of the respective variables.
	In the following, let $j_i' = \lceil j_i /\log q \rceil$ for $i \in [3]$.
	Note that there is precisely one truth assignment of the variables $x_{j_1},
	x_{j_2}$, and $x_{j_3}$ that does not satisfy $C_j$. 
	Choose $\ell_1, \ell_2, \ell_3 \in \{0,1\}$ such that $\ell_i = 0$ if and only if the $i$-th
	variable in $C_j$ appears negated. 
	For~$i \in [3]$ let $F_i \subseteq [q]$ be those colors
	whose binary expansion differs from $\ell_i$ at the $(j_i \bmod (\log q))$-th least significant bit, and define~$F^{+q}_i \defeq \{q + c \mid c \in F_i\}$ and~$F^{+2q}_i \defeq \{2q + c \mid c \in F_i\}$. This implies that the truth assignment encoded by a proper
	list-coloring of~$\cV$ falsifies the $i$-th literal of $C_j$ if
	and only if it uses a color from $F_i$ on vertex $v_{1, j'_i}$. By Claim~\ref{obsNoUCVCDelta}, this happens if and only if it uses a color from~$F^{+q}_i$ on vertex~$v_{2, j'_i}$, which happens if and only if it uses a color of~$F^{+2q}_i$ on vertex~$v_{3,j'_i}$.
	Hence the assignment encoded by a proper list-coloring satisfies clause $C_j$ if
	and only if the colors appearing on $(v_{1, j_1'}, v_{2, j_2'}, v_{3, j_3'})$
	do not belong to the set $F_1 \times F^{+q}_2 \times F^{+2q}_3$. To encode the requirement
	that $C_j$ be satisfied into the graph $G_{3q}$, for each $(\gamma_1, \gamma_2,
	\gamma_3) \in F_1 \times F^{+q}_2 \times F^{+2q}_3$ we add a vertex $w_{\gamma_1,
	\gamma_2, \gamma_3}$ to $G_{3q}$ that is adjacent to $v_{1, j_1'}, v_{2,
	j_2'}$, and $v_{3, j_3'}$ and whose list is $\{\gamma_1, \gamma_2, \gamma_3\}$. 
	The threefold redundancy we incorporated ensures that the three colors in each forbidden triple are all distinct. Therefore, if one of the three neighbors of~$w_{\gamma_1, \gamma_2, \gamma_3}$ does not receive its forbidden color, then~$w_{\gamma_1, \gamma_2, \gamma_3}$ can properly receive that color. This would not hold if there could be duplicates among the forbidden colors.
	\begin{figure}
		\centering
		\includegraphics[width=.9\textwidth]{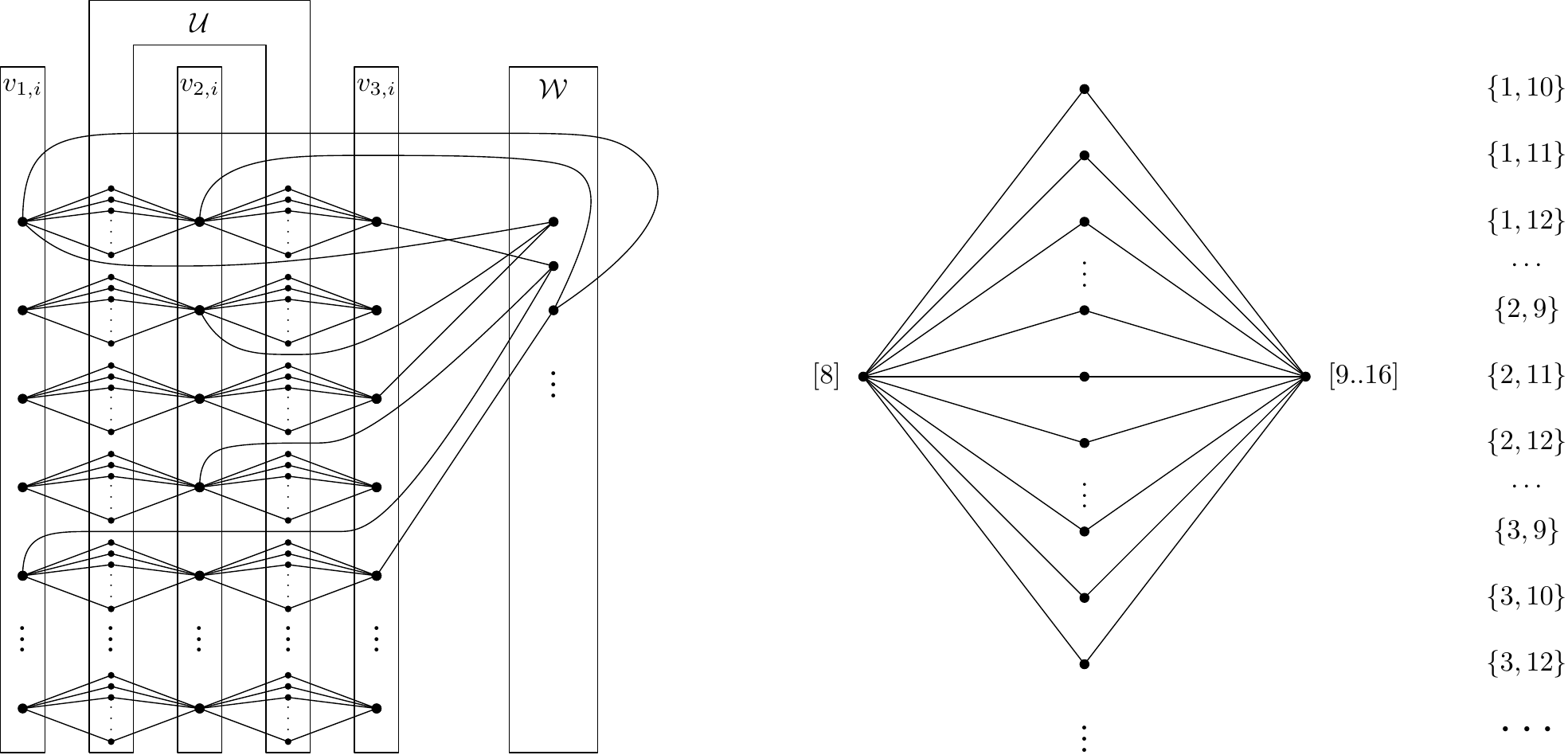}
		\caption{An illustration of the reduction given in the proof of Theorem
		\ref{thmNoUCQColETH}.
		On the left there is a schematic overview and on the right an example of a 
		subgraph induced by two vertices~$v_{1, j}$ and~$v_{2, j}$ together with the corresponding vertices in~$\cU$ for \textProb{24-List-Coloring} (where the lists of the vertices in the middle are displayed to their right).}
		\label{figNoUCQColETHRed}
	\end{figure}
	The reduction is finished by adding these vertices for each clause $C_j \in \varphi$. We denote the set of clause vertices by~$\cW$. For an illustration see Figure \ref{figNoUCQColETHRed}.
	\begin{claim}\label{propNoUCCorr}
		The formula $\varphi$ has a satisfying assignment, if and only if the graph
		$G_{3q}$ obtained via the above reduction is $3q$-list-colorable.
	\end{claim}
	\begin{claimproof}
Suppose $\varphi$ has a satisfying assignment $\psi \colon [n] \to \{0,1\}$. Let~$\gamma_{\psi}$ be the corresponding proper coloring of~$\cV$, as described above. 
	We argue that $\gamma_\psi$ can be extended to the vertices $\cW$ as well. 
	Let $C_j \in \varphi$ be a clause on variables $x_{j_1}, x_{j_2}$, and $x_{j_3}$
	and let $w_{\gamma_1, \gamma_2, \gamma_3} \in \cW$ be a vertex we introduced in
	the construction above for $C_j$. For $i \in [3]$, let $\gamma_\psi^i =
	\gamma_\psi(v_{i, \lceil j_i / \log q\rceil })$.
	
	Since $\gamma_\psi$ encodes a satisfying assignment, we know that there exists
	an $i^* \in [3]$, such that $\gamma_\psi^{i^*} \neq \gamma_{i^*}$ (since
	otherwise, $\psi$ is not a satisfying assignment to $\varphi$). Hence, the
	color $\gamma_{i^*}$ is not blocked from the list of vertex $w_{\gamma_1,
	\gamma_2, \gamma_3}$ which can then be properly colored. By Claim
	\ref{obsNoUCVCDelta} we know that the remaining vertices $\cU$
	can be properly list-colored as well.
	\par
	Conversely, suppose that $G_{3q}$ is properly list-colored. We show that each
	proper coloring must correspond to a truth assignment that satisfies
	$\varphi$. For the sake of a contradiction, suppose that there is a proper
	list-coloring $\gamma_\psi \colon V(G) \to [3q]$ which encodes a truth assignment
	$\psi$ that does not satisfy $\varphi$. Let $C_j \in \varphi$ denote a
	clause which is not satisfied by $\psi$ on variables $x_{j_1}, x_{j_2}$, and
	$x_{j_3}$. For $i \in [3]$, we denote by $\gamma_\psi^i = \gamma_\psi(v_{i,
	\lceil j_i / \log q \rceil })$ the colors of the variable vertices encoding the
	truth assignment of the variables in $C_j$. Since $\psi$ does not satisfy $C_j$
	we know that we added a vertex $w_{\gamma_\psi^1, \gamma_\psi^2,
	\gamma_\psi^3}$ to $\cW$, which is adjacent to $v_{1, \lceil j_1 / \log q
	\rceil}$, $v_{2, \lceil j_2 / \log q \rceil}$, and $v_{3, \lceil j_3 / \log q
	\rceil}$. This means that the colors $\gamma_\psi^1, \gamma_\psi^2$, and
	$\gamma_\psi^3$ appear on a vertex which is adjacent to $w_{\gamma_\psi^1,
	\gamma_\psi^2, \gamma_\psi^3}$ and hence the coloring $\gamma_\psi$ is
	improper, a contradiction.
\end{claimproof}
	We have shown how to reduce an instance of \textProb{3-SAT} to an instance of
	\textProb{$3q$-List-Coloring}. We modify the graph $G_{3q}$ to obtain an
	instance of \textProb{$q$-Coloring} which preserves the correctness of the
	reduction. We add a clique $K_{3q}$ of $3q$ vertices to $G_{3q}$, each of whose
	vertices represents one color. We make each vertex in $v \in \cV \cup
	\cW \cup \cU$ adjacent to each vertex in $K_{3q}$ that represents a color
	which does not appear on~$v$'s list in the list-coloring instance. (The same
	trick was used in the proof of Theorem 6.1 in \cite{LMS11}.) It follows that the graph without~$K_{3q}$ has a proper list-coloring 
	if and only if the new graph has a proper $3q$-coloring.
	\par
	We now compute the size of $G_{3q}$ in terms of $n$ and $q$ and give a bound
	on the size of a vertex cover of $G_{3q}$. We observe that $|\cV| = 3 \lceil
	n/\log q \rceil$, $|\cU| = \cO(q^2 \cdot \lceil n / \log q \rceil)$, and
	clearly, $|V(K_{3q})| = 3q$. To bound the size of $\cW$, we observe that for
	each clause $C_j$, we added $(2^{\log q - 1})^3$ vertices (since we considered
	all triples of bitstrings of length $\log q$ where one character is fixed in
	each string) and hence $|\cW| = \cO(q^3 \cdot m)$ with~$m$ the number of clauses in~$\phi$.
	It is easy to see that $\cV \cup V(K_{3q})$ is a vertex cover of $G_{3q}$ and
	hence $G_{3q}$ has a vertex cover of size $3 \lceil n/\log q \rceil + 3q$.
	\par
	Assuming there is an algorithm that solves \textProb{$q$-Coloring} on
	$\mbox{\textsc{Independent}}+kv$ graphs in time $\cO^*(\theta^k)$ together with an
	application of the above reduction (whose correctness follows from Claim
	\ref{propNoUCCorr}) would yield an algorithm for \textProb{3-SAT} that runs in
	time
	\begin{align*}
		&\theta^{3 \lceil n/\log q \rceil + 3q}\cdot ((q^2 + 3) \lceil n/\log q \rceil + 3q
		+ q^3 \cdot m)^{\cO(1)} = \theta^{3 \lceil n/\log q \rceil + 3q}\cdot (n +
		m)^{\cO(1)}
		\\
		= ~&\theta^{3 \lceil n/\log q \rceil + 3q}\cdot n^{\cO(1)}
		= \cO^*\left(\theta^{3 \lceil n/\log q \rceil + 3q}\right) =
		\cO^*\left(2^{\frac{3 \log \theta}{\log q}n}\right).
	\end{align*}
	Hence, for any $\varepsilon > 0$ we can choose a constant $q$ large enough such that $(3 \log \theta)/(\log q) < \varepsilon$ and Theorem \ref{thmNoUCQColETH} follows.
\end{proof}

\subsection{No Nontrivial Runtime Bound for Path+kv
Graphs}\label{secLBSETHLinFor} 
We now strengthen the lower bound for $\mbox{\textsc{Forest}}+kv$ graphs due to
\cite{LMS11} to the more restrictive class of $\mbox{\textsc{Linear Forest}}+kv$
graphs. The key idea in our reduction is that we treat the clause size in a
satisfiability instance as a constant, which allows for constructing a graph of
polynomial size.
The following lemma describes the clause gadget that will be used in the reduction.
\begin{lemma}\label{lemAuxPathSETH}
For each~$q \geq 3$ there is a polynomial-time algorithm that, given~$(c_1, \ldots, c_m) \in [q]^m$, outputs a $q$-list-coloring instance~$(P, \Lambda)$ where~$P$ is a path of size~$\cO(m)$ containing distinguished vertices~$(\pi_1, \ldots, \pi_m)$, such that the following holds. For each~$(d_1, \ldots, d_m) \in [q]^m$ there is a proper list-coloring~$\gamma$ of~$P$ in which~$\gamma(\pi_i) \neq d_i$ for all~$i$, if and only if~$(c_1, \ldots, c_m) \neq (d_1, \ldots, d_m)$.
\end{lemma}
\begin{proof}
The path~$P$ consists of consecutive vertices~$v_0, v_1, \ldots, v_{6m}, v_{6m+1}$. Vertex~$v_0$ is the source and~$v_{6m+1}$ is the sink. The remaining~$6m$ vertices are split into~$m$ groups~$D_1, \ldots, D_m$ consisting of six consecutive vertices~$v_{6(i-1)+1},\ldots,v_{6i}$ ($i \in [m]$) each. We first add some colors to the lists of these vertices which are allowed regardless of~$(c_1, \ldots c_m)$. Later we will add some more colors to the lists of selected vertices to obtain the desired behavior.

Initialize the `default' list of vertex~$v_i$ for~$i \in [6m]$ to contain the two colors~$\{(i \bmod 3) + 1, (i+1 \mod 3) + 1\}$, so that the first few lists are~$\{2,3\}$,~$\{3,1\}$,~$\{1,2\}$, and so on. Initialize~$\Lambda(v_0) \defeq \Lambda(v_{6m+1}) \defeq \{2\}$. With these lists, there is no proper list-coloring of~$P$. The color for the source vertex is fixed to~$2$, forcing the color of~$v_1$ to~$3$, which forces~$v_2$ to~$1$, and generally forces~$v_i$ to color~$(i + 1) \bmod 3 + 1$. Hence~$v_{6m}$ is forced to~$(6m + 1) \bmod 3 + 1 = 2$, creating a conflict with the sink~$v_{6m+1}$ which is also forced to color~$2$. 

We now introduce additional colors on some lists, and identify the distinguished vertices $\pi_1, \ldots, \pi_m$
among the vertices~$v_{i'}$ (where~$i' \in [6m]$), to allow proper list-colorings under the stated conditions. (Note that in the rest of the proof, we will make use of two symbols for any distinguished vertex, depending on which is more convenient at the time:~$\pi_i$ where~$i \in [m]$ and~$v_{i'}$ where $i' \in [6m]$.)
For a group~$D_i$ of six consecutive vertices, the \emph{interior} of the group consists of the middle four vertices. 
For each index~$i \in [m]$, choose~$\pi_i$ as a vertex from the interior of group~$D_i$ such that~$c_i$ is not on the default list of colors for~$\pi_i$. Since there is no color that appears on all of the default lists of the four interior vertices, this is always possible. Add~$c_i$ to the list of allowed colors for~$\pi_i$. This completes the construction of the list-coloring instance~$(P,\Lambda)$. For an illustration see Figure \ref{figSethPath}.

\begin{figure}
	\centering
	\includegraphics[width=.9\textwidth]{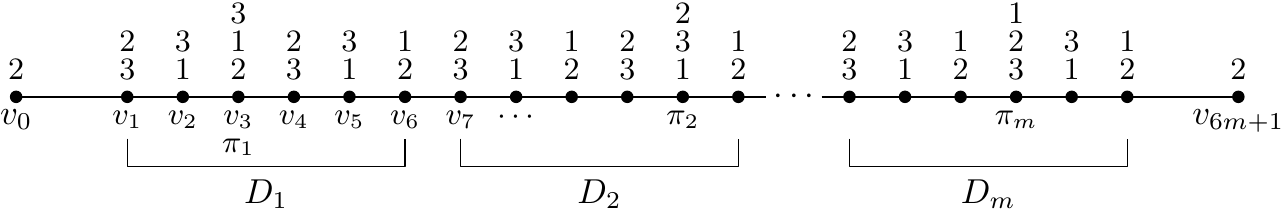}
	\caption{An example \textProb{3-List-Coloring} instance created as in the proof of Lemma \ref{lemAuxPathSETH}, where~$(c_1,\ldots, c_m) = (3,2,\ldots, 1)$.}
	\label{figSethPath}
\end{figure}

It is easy to see that the construction can be performed in polynomial time. To conclude the proof, we argue that~$(P,\Lambda)$ has the desired properties. Observe that if~$(d_1, \ldots, d_m) = (c_1, \ldots, c_m)$, then a proper list-coloring~$\gamma$ of~$P$ in which~$\gamma(\pi_i) \neq d_i = c_i$ for all~$i \in [m]$ would in fact be a proper list-coloring of~$P$ under the default lists before augmentation, which is impossible as we argued earlier. It remains to argue that when~$(d_1, \ldots, d_m)$ differs from~$(c_1, \ldots, c_m)$ in at least one position, then~$P$ has a proper list-coloring~$\gamma$ with~$\gamma(\pi_i) \neq d_i$ for all~$i \in [m]$. 
To construct such a list-coloring, for each index~$i \in [m]$ with~$c_i \neq d_i$, assign vertex~$\pi_i$ the color~$c_i$. Since the vertices~$\pi_i$ are interior vertices of their groups, the distinguished vertices are pairwise nonadjacent and this does not result in any conflicts. For distinguished vertices~$\pi_i$ with~$c_i = d_i$, we will assign~$\pi_i$ a color from the default list of vertex~$\pi_i$; since~$c_i$ is not on the default list this results in the desired color-avoidance. We therefore conclude by verifying that the remaining vertices can be assigned a proper color from their default list.

To do so, assign the source vertex its forced color and propagate the coloring as described above, until we reach the first distinguished vertex~$\pi_{i}$ with~$c_{i} \neq d_{i}$ (where~$i \in [m]$).  
Let~$i' \in [6m]$ denote the index of~$\pi_i$ among all vertices of $P$, i.e.~$\pi_i = v_{i'}$. In the current partial coloring,~$v_{i' - 1}$ received color~$((i' - 1) + 1) \bmod 3 + 1 = i' \bmod 3 + 1$ which is a color on the default list of~$v_{i'}$. Hence, we do not create a conflict between vertices~$v_{i' - 1}$ and~$v_{i'}$ as we gave~$v_{i'}$ the color~$c_i$ which was not on~$v_{i'}$'s default list by construction. The other color on the default list of~$v_{i'}$ is~$(i'+1) \bmod 3 + 1$, which is also on the list of~$v_{i' + 1}$, as~$\Lambda(v_{i'+1}) = \{(i'+1) \bmod 3 + 1, (i'+2) \bmod 3 + 1\}$. Hence, assigning~$v_{i'+1}$ color~$(i'+1) \bmod 3 + 1$ does not create a conflict between~$v_{i'}$ and~$v_{i'+1}$, again since we assigned $v_{i'}$ a color which was not on its default list. 

\begin{itemize}
	\item If~$i$ was the last index for which~$c_i \neq d_i$, then, for all $i'' \in [(i' + 2)..6m]$ we continue giving vertex~$v_{i''}$ color~$i'' \bmod 3 + 1$. This way the sink can be properly list-colored. 
	\item If not, we give~$v_{i' + 2}$ color~$i' \bmod 3 + 1~(= (i' + 3) \bmod 3 + 1)$. Note that since all distinguished vertices are interior vertices of the groups, $v_{i'+2}$ cannot be a distinguished vertex and hence has not been previously assigned a color. We now propagate this coloring along the path as before until we reach the next distinguished vertex which has already been assigned a color.
\end{itemize}
We repeat the construction until all vertices are properly list-colored.
\end{proof}

\begin{theorem}\label{thmQLColLinFor}	
	For any $\varepsilon > 0$ and constant $q \ge 3$, \textProb{$q$-Coloring} on $\mbox{\textsc{Linear Forest}} + kv$ graphs cannot be solved in time $\cO^*((q - \varepsilon)^k)$, unless $\SETH$ fails.
\end{theorem}
\begin{proof}
	To prove the theorem, we will first show that if \textProb{$q$-List-Coloring}
	on $\mbox{\textsc{Linear Forest}} + kv$ graphs can be solved in time
	$\cO^*((q-\varepsilon)^k)$ for $q \ge 3$ and some $\varepsilon > 0$, then
	\textProb{$s$-SAT} can be solved in time $\cO^*((2-\delta)^n)$ for some $\delta
	> 0$ and any $s \in \cO(1)$, contradicting $\SETH$. By the same argument as in
	the proof of Theorem \ref{thmNoUCQColETH}, we then extend the lower bound to \textProb{$q$-Coloring}.
	\par
	Suppose we have an instance $\varphi$ of \textProb{$s$-SAT} on variables
	$x_1,\ldots,x_n$. We construct a graph $G$ and lists $\Lambda \colon V(G) \to [q]$,
	such that $G$ is properly list-colorable if and only if $\varphi$ is
	satisfiable.
	The first part of the reduction is inspired by the reduction of Lokshtanov et al.~\cite[Theorem 6.1]{LMS11}, which we repeat here for completeness.
		We choose an integer constant $p$ depending on $q$ and $\varepsilon$ 
		and group the variables of
		$\varphi$ into $t$ groups $F_1,\ldots,F_t$ of size $\lfloor \log q^p\rfloor$
		each. We call a truth assignment for the variables in $F_i$ a \emph{group
		assignment}. We say that a group assignment satisfies clause $C_j \in \varphi$
		if $C_j$ contains at least one literal which is set to \textsc{True} by the
		group assignment. For each group $F_i$, we add a set of $p$ vertices
		$v_i^1,\ldots,v_i^p$ to $G$, in the following denoted by $\cV_i$ with
		$\Lambda(v_i^j) = [q]$ for all $i$ and $j$. Each coloring of the vertices
		$\cV_i$ will encode one group assignment of $F_i$. 
		We fix some efficiently computable
		injection 
		$f_i \colon \{0,1\}^{|F_i|} \to [q]^p$ that assigns to each group assignment for $F_i$ a
		distinct $p$-tuple of colors. This is possible since there are $q^p \ge 2^{|F_i|}$ possible colorings of~$p$ vertices.
		For a variable $x_i \in \varphi$ we can identify the 
		set of vertices whose colorings encode the assignment of the group containing $x_i$. Since each group has size $\lfloor \log q^p \rfloor$, the
		truth assignments of a variable $x_i \in \varphi$ are encoded by (some) 
		colorings of the vertices in $\cV_{i'}$, where
		$i' = \lceil i/\lfloor \log q^p \rfloor \rceil$.

We now construct the main part of the graph $G$. Let $C_j \in \varphi$ be a clause on variables $x_{j_1},\ldots,x_{j_{s'}}$, where $s' \in [s]$. The truth assignments of these variables are encoded by the colorings of the vertices in $\cV_{C_j} = \bigcup_{i \in [s']}\cV_{\lceil j_i / \lfloor \log q^p \rfloor \rceil}$. We say that a coloring~$\mu \colon \cV_{C_j} \to [q]$ is a \emph{bad} coloring for~$C_j$ if there is a group for which the coloring does not represent a group assignment, or if the group assignments encoded by~$\mu$ do not satisfy clause~$C_j$. 

For each bad coloring $\mu$ we construct a path using Lemma \ref{lemAuxPathSETH} which ensures that $G$ is not properly list-colorable if~$\mu$ appears on~$\cV_{C_j}$. Let~$j_i' = \lceil j_i / \lfloor \log q^p \rfloor \rceil$ and consider the following vector of colors induced by $\mu$:
\begin{align}
c_\mu = \left(\mu\left(v_{j_1'}^1\right), \ldots, \mu\left(v_{j_1'}^p\right), \ldots, \mu\left(v_{j_{s'}'}^1\right), \ldots, \mu\left(v_{j_{s'}'}^p\right)\right) \label{eqColVectRep}
\end{align}
We add to $G$ a path $P_{c_\mu}$ constructed according to Lemma \ref{lemAuxPathSETH} with $c_\mu$ as the input vector of colors. Let~$(\pi_1,\ldots,\pi_{p\cdot s'})$ denote the distinguished vertices of $P_{c_\mu}$. We make each variable vertex~$v_{j_i'}^\ell \in \cV_{C_j}$ (where~$i \in [s']$ and~$\ell \in [p]$) adjacent to the distinguished vertex~$\pi_{p\cdot (i-1) + \ell}$ in $P_{c_\mu}$, intending to ensure that if all vertices in~$\cV_{C_j}$ are colored according to~$\mu$, then this partial list-coloring on~$G$ cannot be extended to~$P_{c_\mu}$.
Adding such a path for each clause in~$\varphi$ and each bad coloring finishes the construction of~$(G, \Lambda)$.

We first count the number of vertices in~$G$ and then prove the correctness of the reduction. There are~$\cO(n)$ variable vertices and for each of the~$m$ clauses, there are at most~$q^{p\cdot s}$ bad colorings, each of which adds a path on at most~$\cO(p \cdot s)$ vertices to $G$, by Lemma \ref{lemAuxPathSETH}. Hence, the number of vertices in $G$ is at most
\begin{align}
	\cO\left(n + m \cdot q^{p \cdot s}(p\cdot s)\right) = \cO(n + m) = n^{\cO(1)}, \label{eqSETHLinFor}
\end{align}
as~$p, q, s \in \cO(1)$ and~$m = \cO(n^s)$.

\begin{claim}\label{claimSethPathCorrectness}
	$(G, \Lambda)$ is properly~$q$-list-colorable if and only if~$\varphi$ has a satisfying assignment.
\end{claim}
\begin{claimproof}
Suppose~$\varphi$ has a satisfying assignment~$\psi$. For each group~$\cV_i$ the assignment~$\varphi$ dictates a group assignment, which corresponds to a coloring on~$\cV$ by the chosen injection~$f_i$. Let~$\gamma_\psi \colon \bigcup_i \cV_i \to [q]$ denote the coloring of the variable vertices that encodes~$\psi$. We argue that~$\gamma_\psi$ can be extended to the rest of~$G$, respecting the lists~$\Lambda$. For every~$C_j \in \varphi$ on variables~$x_{j_1},\ldots,x_{j_{s'}}$ and every bad coloring~$\mu \colon \bigcup_{i = 1}^{s'} \cV_{j_i'} \to [q]$ w.r.t.~$C_j$ (where~$j_i' = \lceil j_i / \lfloor \log q^p \rfloor \rceil$), we added a path~$P_{c_\mu}$ to~$G$, constructed according to Lemma \ref{lemAuxPathSETH}, whose distinguished vertices we denote by~$(\pi_1, \ldots, \pi_{p\cdot s'})$. Note that $c_\mu$ denotes the vector representation of the coloring $\mu$ as in (\ref{eqColVectRep}). Let~$c_\gamma$ denote the vector representation of~$\gamma$ restricted to the variable vertices~$\cup_{i=1}^{s'} \cV_{j_i'}$, appearing in the same order as in~$c_\mu$. Since~$\gamma_\psi$ encodes a satisfying assignment of~$\varphi$, $c_\mu \neq c_\gamma$. Hence, by Lemma \ref{lemAuxPathSETH}, we can extend~$\gamma_\psi$ to~$P_{c_\mu}$ without creating a conflict; it asserts that there is a proper list-coloring $\gamma'$ on $P_{c_\mu}$ such that~$\gamma(v_{j_i'}^\ell) = c_\gamma(p \cdot (i-1) + \ell) \neq \gamma'(\pi_{p\cdot(i-1) + \ell})$ for all~$i \in [s']$ and~$\ell \in [p]$. Hence, every pair of adjacent vertices between the vertices of~$P_{c_\mu}$ and the vertices encoding the truth assignments of the variables in~$C_j$ can be list-colored properly and we can conclude that~$\gamma_\psi$ can be extended to~$P_{c_\mu}$ and subsequently, to all of~$G$.

Now suppose~$(G, \Lambda)$ has a proper list-coloring~$\gamma$ and assume for the sake of a contradiction that $\varphi$ does not have a satisfying assignment. Then, the restriction of any list-coloring of~$G$ to (some of) the variable vertices~$\bigcup_i \cV_i$ must be a bad coloring for some clause in~$\varphi$. Let~$C_j$ denote such a clause for $\gamma$ and let $c_\gamma$ denote the corresponding vector of colors, restricted to the variable vertex groups that encode the truth assignments to the variables in $C_j$. We added a path~$P_{c_\gamma}$ to $G$ which by Lemma \ref{lemAuxPathSETH} cannot be properly list-colored such that each distinguished vertex gets a color which is different from the color of the variable vertex it is adjacent to. Hence, one of the distinguished vertices of $P_{c_\gamma}$ creates a conflict and we have a contradiction.
\end{claimproof}

Since $G$ consists of the variable vertices attached to a set of disjoint paths, we observe the following.
\begin{observation}\label{obsQLColLinFor}
		$\bigcup_i \cV_i$ is a modulator to \textsc{Linear Forest}.
	\end{observation}
	The previous observation can easily be verified, since $G$ consists of the variable vertices attached to a set of disjoint paths. By Claim \ref{claimSethPathCorrectness} and Observation \ref{obsQLColLinFor} we can
	now finish the proof in the same way as the proof of \cite[Theorem 6.1]{LMS11}, in particular Lemma
	6.4 yields the claim.
	\begin{claim}[Cf.\ Lemma 6.4 in \cite{LMS11}]\label{claim:seth:path:runtime}
		If \textProb{$q$-List-Coloring} on $\mbox{\textsc{Linear Forest}}+kv$ graphs
		can be solved in time $\cO^*((q-\varepsilon)^k)$ for some $\varepsilon < 1$,
		then \textsc{$s$-SAT} can be solved in $\cO^*((2-\delta)^n)$ time, for some
		$\delta < 1$ and any $s \in \cO(1)$.
	\end{claim}
	\begin{claimproof}
		Let~$\lambda \defeq \log_q(q - \varepsilon)^k < 1$, such that~$(q-\varepsilon)^k = q^{\lambda k}$.
		Note that by (\ref{eqSETHLinFor}), the size
		of $G$ is polynomial in $n$, the number of variables of $\varphi$.
		We choose a sufficiently large $p$ such that $\delta' = \lambda \frac{p}{p-1}
		< 1$. Given an instance $\varphi$ of \textProb{$s$-SAT}, we use the above
		reduction to obtain $(G, \Lambda)$, an instance of \textProb{$q$-List-Coloring}.
		Correctness follows from Claim \ref{claimSethPathCorrectness}. By Observation
		\ref{obsQLColLinFor} we know that $G$ has a modulator to \textsc{Linear
		Forest} of size $p \lceil \frac{n}{\lfloor p \log q \rfloor} \rceil$. By the
		choice of $p$ we have
			$\lambda p \lceil \frac{n}{\lfloor p \log q \rfloor} \rceil \le \lambda p
			\frac{n}{(p-1)\log q} + \lambda p \le \delta' \frac{n}{\log q} + \lambda p$. Hence, \textProb{$s$-SAT} can be solved in~$\cO^*(2^{\delta' n + \lambda p}) = \cO^*(2^{\delta' n}) = \cO^*((2 - \delta)^n)$ time for some $\delta > 0$ which does not depend on $s$.
\end{claimproof}
	We have given a reduction from \textProb{$s$-SAT} to
	\textProb{$q$-List-Coloring} on $\mbox{\textsc{Linear Forest}}+kv$ graphs. As
	in the proof of Theorem \ref{thmNoUCQColETH}, we can make the reduction
	work for \textProb{$q$-Coloring} as well by adding a clique $K_q$ of $q$
	vertices to the graph, each of which represents one color and then making each
	vertex in $G$ adjacent to each vertex in $K_q$ which represents a color that is
	not on its list.
	Since this increases the size of the modulator by $q$, which is a
	constant, this does not affect asymptotic runtime bounds and completes the 
	proof of Theorem \ref{thmQLColLinFor}.
\end{proof}

Note that we can modify the reduction in the proof of Theorem
\ref{thmQLColLinFor} to give a lower bound for $\mbox{\textsc{Path}}+kv$ graphs
as well: We simply connect all paths that we added to the graph to one long path,
adding a vertex with a full list between each pair of adjacent paths. 
\begin{corollary}\label{corqColPathLB}
	For any $\varepsilon > 0$ and constant $q \ge 3$, \textProb{$q$-Coloring} on $\mbox{\textsc{Path}}+kv$ graphs cannot be solved in time $\cO^*((q-\varepsilon)^k)$, unless $\SETH$ fails.
\end{corollary}

\section{A Tighter Treedepth Boundary}\label{secTD}
In Lemma \ref{lemQLColUBTDepth} we showed that if the $(q+1)$-colorable members
of a hereditary graph class $\cF$ have bounded treedepth, then $\cF$ has
constant-size \textsc{No}-certificates for \textProb{$q$-List-Coloring} and hence $\cF+kv$ has
nontrivial algorithms for \textProb{$q$-(List-)Coloring} parameterized by the
size of a given modulator to $\cF$. One might wonder whether a graph class $\cF+kv$ has
nontrivial algorithms for \textProb{$q$-Coloring} parameterized by a given
modulator to $\cF$ \emph{if and only if} all $(q+1)$-colorable members in $\cF$ have bounded
treedepth. However, this is not the case. In \cite[Lemma 4]{JK13} the authors
showed that \textProb{$q$-Coloring} parameterized by the size of a modulator to
the class \textsc{Cograph} has nontrivial algorithms. Clearly, complete
bipartite graphs are cographs and it is easy to see that (the 2-colorable
balanced biclique) $K_{n, n}$ has treedepth $n+1$. In this section we show that,
unless $\SETH$ fails, bicliques are in some sense the only obstruction to this
treedepth boundary.
\par
We use a combinatorial theorem which in combination with
Corollary \ref{corqColPathLB} will yield the result.
\begin{theorem}[Corollary 3.6 in \cite{JM15}, Theorem 1 in \cite{ALR12}]\label{thm1ALR12}
	For any $s, k \in \bN$ there is a $P(s, k) \in \bN$ such that any graph
	with a path of length $P(s, k)$ either contains an induced path of length $s$, 
	or a $K_k$ subgraph, or an induced $K_{k, k}$ subgraph.
\end{theorem}
\begin{theorem}
	Let $\cF$ be a hereditary class of graphs for which there exists a 
	$t \in \mathbb{N}$ 
	such that $K_{t, t}$ is not contained in $\cF$, let~$q \geq 3$, and suppose $\SETH$
	is true. Then, \textProb{$q$-Coloring} parameterized by a given vertex modulator to
	$\cF$ of size $k$ has $\cO^*((q-\varepsilon)^k)$ time algorithms for some $\varepsilon >
	0$, if and only if all $(q+1)$-colorable graphs in $\cF$ have bounded treedepth.
\end{theorem}
\begin{proof}
Assume the stated conditions hold for~$\cF$ and~$t$. In one direction, if all the~$(q+1)$-colorable graphs in~$\cF$ have their treedepth bounded by a constant, then there are constant-size \textsc{No}-certificates for \textProb{$q$-List-Coloring} on~$\cF$ by Lemma~\ref{lemQLColUBTDepth}, implying the existence of nontrivial algorithms by Theorem~\ref{thmAlgQListColNoCert}.

For the other direction, suppose that there is no finite bound on the treedepth of~$(q+1)$-colorable graphs in~$\cF$. We claim that~$\cF$ contains all paths, which will prove this direction using Corollary~\ref{corqColPathLB}. If the longest (simple) path in a graph~$G$ has length~$k$, then~$G$ has treedepth at most~$k$ since any depth-first search tree forms a valid treedepth decomposition, and has depth at most~$k$ since all its root-to-leaf paths are paths in~$G$. Hence a graph of treedepth more than~$n$ contains a path of length more than~$n$. Since the $(q+1)$-colorable graphs in~$\cF$ have arbitrarily large treedepth, the preceding argument shows that for any~$n$ there is a $(q+1)$-colorable graph in~$\cF$ containing a path of length more than~$n$. In particular, for any~$n$ there is a $(q+1)$-colorable graph~$G_n$ in~$\cF$ containing a (not necessarily induced) path of length~$P(n, \max(t, q+2))$, the Ramsey number of Theorem~\ref{thm1ALR12}. Hence graph~$G_n$ contains an induced path of length~$n$, a clique of size~$\max(t, q+2)$, or an induced biclique with sets of size~$\max(t,q+2)$. Since a~$(q+2)$-clique is not~$(q+1)$-colorable,~$G_n$ contains no such clique. If~$G_n$ contains an induced biclique subgraph with sets of size~$\max(t,q+2)$, then since~$\cF$ is hereditary it would contain~$K_{t,t}$, which contradicts our assumption on~$\cF$. Hence~$G_n$ contains an induced path of length~$n$, implying that~$\cF$ contains the induced path of length~$n$ since it is hereditary. As this holds for all~$n$, class~$\cF$ contains all paths, implying by Corollary~\ref{corqColPathLB} and SETH that there are no nontrivial algorithms for \textProb{$q$-List-Coloring} parameterized by the size of a given vertex modulator to~$\cF$.
\end{proof}

\section{Conclusion}\label{secConc}
In this paper we have presented a fine-grained parameterized complexity analysis
of the \textProb{$q$-Coloring} and \textProb{$q$-List-Coloring} problems.
We showed that if a graph class $\cF$ has \textsc{No}-certificates for
\textProb{$q$-List-Coloring} of bounded size or if the $(q+1)$-colorable members
of $\cF$ (where $\cF$ is hereditary) have bounded treedepth, then there is an algorithm that solves
\textProb{$q$-Coloring} on graphs in $\cF+kv$ (graphs with vertex modulators of
size $k$ to $\cF$) in time $\cO^*((q-\varepsilon)^k)$ for some $\varepsilon >
0$ (depending on $\cF$). The parameter treedepth revealed itself as a boundary
in some sense:
We showed that $\mbox{\textsc{Path}}+kv$ graphs do not have $\cO^*((q-\varepsilon)^k)$ time
algorithms for any $\varepsilon > 0$ unless $\SETH$ is false --- and paths are
arguably the simplest graphs of unbounded treedepth. Furthermore we proved 
that if a graph class $\cF$ does not have large bicliques, then $\cF+kv$ graphs 
have $\cO^*((q-\varepsilon)^k)$ time algorithms, for some $\varepsilon > 0$, if and 
only if $\cF$ has bounded treedepth.
\par 
Treedepth is an interesting graph parameter which in many cases also allows for
polynomial space algorithms where e.g.\ for treewidth this is typically
exponential.
It would be interesting to see how the problems studied by Lokshtanov et al.\
\cite{LMS11} behave when parameterized by treedepth. Naturally, a fine-grained
parameterized complexity analysis as we did might be interesting for other
problems as well.
\begin{openProblem*}
	Consider a different problem than \textProb{$q$-Coloring}, for example another
	problem studied in \cite{LMS11}. For which parameters in the hierarchy can we
	improve upon the base of the exponent of the $\SETH$-based lower bound? Does
	the parameter treedepth establish a diving line in this sense as well?
\end{openProblem*}

\bibliographystyle{plain}
\bibliography{References_pc_chromatic}

\begin{thebibliography}{10}

\bibitem{ALR12}
Aistis Atminas, Vadim~V. Lozin, and Igor Razgon.
\newblock Linear time algorithm for computing a small biclique in graphs
  without long induced paths.
\newblock In {\em SWAT}, pages 142--152. Springer, 2012.

\bibitem{BHK09}
Andreas Bj{\"o}rklund, Thore Husfeldt, and Mikko Koivisto.
\newblock Set partitioning via inclusion-exclusion.
\newblock {\em SIAM Journal on Computing}, 39(2):546--563, 2009.
\newblock Based on two extended abstracts appearing in FOCS '06.

\bibitem{Cai03}
Leizhen Cai.
\newblock Parameterized complexity of vertex colouring.
\newblock {\em Discrete Applied Mathematics}, 127(3):415--429, 2003.

\bibitem{CFKLMPPS15}
Marek Cygan, Fedor~V. Fomin, {\L}ukasz Kowalik, Daniel Lokshtanov, Daniel Marx,
  Marcin Pilipczuk, Micha{\l} Pilipczuk, and Saket Saurabh.
\newblock {\em Parameterized Algorithms}.
\newblock Springer, 1st edition, 2015.

\bibitem{Die10}
Reinhard Diestel.
\newblock {\em Graph Theory}, volume 173 of {\em Graduate Texts in
  Mathematics}.
\newblock Springer, 4th edition, 2010.
\newblock Corrected reprint 2012.

\bibitem{DF13}
Rodney~G. Downey and Michael~R. Fellows.
\newblock {\em Fundamentals of Parameterized Complexity}.
\newblock Texts in Computer Science. Springer, 2013.

\bibitem{FJR13}
Michael~R. Fellows, Bart M.~P. Jansen, and Frances Rosamond.
\newblock Towards fully multivariate algorithmics: {P}arameter ecology and the
  deconstruction of computational complexity.
\newblock {\em European Journal of Combinatorics}, 34(3):541--566, 2013.
\newblock Previously app. in IWOCA '09.

\bibitem{FG06}
J\"{o}rg Flum and Martin Grohe.
\newblock {\em Parameterized Complexity Theory}, volume {XIV} of {\em Texts in
  Theoretical Computer Science}.
\newblock Springer, 2006.

\bibitem{IP01}
Russel Impagliazzo and Ramamohan Paturi.
\newblock On the complexity of $k$-sat.
\newblock {\em Journal of Computer and System Sciences}, 62(2):367--375, 2001.

\bibitem{IPZ01}
Russel Impagliazzo, Ramamohan Paturi, and Francis Zane.
\newblock Which problems have strongly exponential complexity?
\newblock {\em Journal of Computer and System Sciences}, 63(4):512--530, 2001.

\bibitem{JK13}
Bart M.~P. Jansen and Stefan Kratsch.
\newblock Data reduction for graph coloring problems.
\newblock {\em Information \& Computation}, 231:70--88, 2013.
\newblock Previously app. in FCT '11.

\bibitem{JM15}
Bart M.~P. Jansen and D\'{a}niel Marx.
\newblock Characterizing the easy-to-find subgraphs from the viewpoint of
  polynomial-time algorithms, kernels and turing kernels.
\newblock {\em Ar{X}iv e-prints}, 2015.
\newblock arxiv:1410.0855, also app. in SODA '15 pp. 616-629.

\bibitem{LMS11}
Daniel Lokshstanov, D{\'a}niel Marx, and Saket Saurabh.
\newblock Known algorithms for graphs of bounded treewidth are probably
  optimal.
\newblock In {\em SODA}, pages 777--789. SIAM, 2011.

\bibitem{NO12}
Jaroslav Ne\v{s}et\v{r}il and Patrice Ossona~de Mendez.
\newblock {\em Sparsity. Graphs, Structures and Algorithms}, volume~28 of {\em
  Algorithms and Combinatorics}.
\newblock Springer, 2012.

\bibitem{Niedermeier06}
Rolf Niedermeier.
\newblock {\em Invitation to Fixed-Parameter Algorithms}.
\newblock Oxford University Press, 2006.

\bibitem{Vas15}
Virginia~Vassilevska Williams.
\newblock Hardness of easy problems: Basing hardness on popular conjectures
  such as the strong exponential time hypothesis.
\newblock In {\em IPEC}, volume~43 of {\em LIPIcs}, pages 16--28, 2015.

\end{thebibliography}

\end{document}